\documentclass[10pt]{article}
\def\N{\mathbb{N}}
\def\Z{\mathbb{Z}}

\def\K{\mathbb{K}}

\newcommand {\dd}{\,{\rm div}\,}

\newcommand {\ord}{\mathop{\rm ord}}

\newcommand {\cB} {{\cal B}}
\newcommand {\cP} {{\cal P}}
\newcommand {\cC} {{\cal C}}
\newcommand {\cL} {{\cal L}_\cB}

\newcommand {\cR} {{\cal R}_\cB}

\newcommand {\cE} {{\cal E}}

\def\KK{\K[[\cB]]}

\newcommand{\qed}{\hfill $\square$ \smallskip}

\newtheorem{theorem}{\bf Theorem}
\newtheorem{corollary}{\bf Corollary}

\newtheorem{proposition}{\bf Proposition}

\newtheorem{remark}{\bf Remark}

\newtheorem{definition}{\bf Definition}
\newtheorem{example}{\bf Example}

\usepackage{amsmath,amssymb,amsfonts,url,relsize}

\date{}

\author
{Marko Petkov\v sek \\
 Faculty of Mathematics and Physics, University of Ljubljana\\
Institute of Mathematics, Physics and Mechanics, Ljubljana\\
marko.petkovsek@fmf.uni-lj.si}
\title
{Definite Sums as Solutions of Linear Recurrences With Polynomial Coefficients}

\begin{document}
\maketitle
\begin{abstract}
We present an algorithm which, given a linear recurrence operator $L$ with polynomial coefficients, $m \in \N\setminus\{0\}$,
$a_1,a_2,\ldots,a_m \in  \N\setminus\{0\}$ and $b_1,b_2,\ldots,b_m \in \K$, returns a linear recurrence operator $L'$ with rational coefficients such that for every sequence $h$,
\[
L\left(\sum_{k=0}^\infty \prod_{i=1}^m \binom{a_i n + b_i}{k} h_k\right) = 0
\]
if and only if $L' h = 0$.
\end{abstract}


\smallskip
MSC (2010)  68W30; 33F10

\section{Introduction}

Holonomic sequences are, by definition, given by a homogeneous linear recurrence with polynomial coefficients (and suitable initial conditions). Often one wishes to find an \textit{explicit representation} of a holonomic sequence, so algorithms have been devised to find solutions of such recurrences within some class of explicitly representable sequences (such as, e.g., polynomial \cite{A89poly}, rational \cite{A89rat, unden}, hypergeometric \cite{hyper}, d'Alembertian \cite{AP94}, or Liouvillian \cite{HS} sequences). These classes do not exhaust all explicitly representable holonomic sequences. For example, every definite hypergeometric sum on which Zeilberger's \textit{Creative Telescoping} algorithm \cite{Zeil90, Zeil91} succeeds is a holonomic sequence, but many such sequences are not Liouvillian. Therefore it makes sense to consider what one might call the \textit{inverse Zeilberger's problem}: given a homogeneous linear recurrence with polynomial coefficients, find its solutions representable as definite sums from a certain class.

Here we make a (modest) first step in this direction by providing an algorithm which, given a linear recurrence operator $L$ with polynomial coefficients and a fixed product of binomial coefficients of the form
\begin{equation}
\label{binom}
F(n,k) \ =\ \prod_{i=1}^m \binom{a_i n + b_i}{k}
\end{equation}
with $a_i \in \N \setminus \{0\}$, $b_i \in \K$, returns a linear recurrence operator $L'$ with rational coefficients such that for any sequence $y$ of the form $y_n = \sum_{k=0}^\infty F(n,k) h_k$, we have $L y = 0$ if and only if $L' h = 0$. This enables us to find all such solutions $y$ where $h$ belongs to a class of holonomic sequences with a known algorithm for converting from recursive to explicit representation (such as those mentioned in the preceding paragraph).  We regard $\sum_{k=0}^\infty F(n,k) h_k$ as a formal series, but note that it is terminating if at least one of the $b_i$ is an integer.

Throughout the paper, $\N = \{0,1,2,\ldots\}$ denotes the set of nonnegative integers, $\K$ a field of characteristic zero, $\K^{\N}$ the set of all sequences with terms from $\K$, $\K[x]$ the $\K$-algebra of univariate polynomials over $\K$, and ${\cal L}_{\K[x]}$ the $\K$-algebra of linear operators $L: \K[x] \rightarrow \K[x]$. 

\begin{definition}
 Let  $m \in \N \setminus \{0\}$ and $j \in \{0,1,\ldots,m-1\}$.
\begin{itemize}
\item A sequence $c \in \K^{\N}$ is called the \emph{$j$-th $m$-section} of $a  \in \K^{\N}$ if $c_k = a_{mk+j}$ for all $k \in \N$. We say that $c$ is obtained from $a$ by \emph{multisection}, and denote it by $s_j^m a$.
\item A sequence $c \in \K^{\N}$ is called the \emph{interlacing} of $a^{(0)},a^{(1)}, \ldots, a^{(m-1)} \in \K^{\N}$ if  $c_k = a^{(k \bmod m)}_{k \dd m}$ for all $k \in \N$. \qed
\end{itemize}
\end{definition}

\section{Formal polynomial series}
 
The power series method is a time-honored approach to solving differential equations by reducing them to  recurrences satisfied by the coefficient sequences of their power series solutions. In \cite{APR} it was shown how, by generalizing the notion of formal \emph{power series} to formal \emph{polynomial series}, one can use this method to find solutions of other linear operator equations such as $q$-difference equations, and recurrence equations themselves, which interests us here. In this section we summarize some relevant definitions, examples and results from \cite{APR}.

\begin{definition}
\rm
A sequence of polynomials $\cB = \langle P_k(x) \rangle_{k=0}^{\infty}$ from $\K[x]$ is a \emph{factorial basis} of $\K[x]$ if for all $k \in \N$:
\begin{enumerate}
\item[\bf P1.] $\deg P_k = k$, 
\item[\bf P2.] $P_k \,|\, P_{k+1}$. \qed
\end{enumerate}
\end{definition}
\begin{definition}
\rm
A factorial basis $\cB$ of $\K[x]$ is \emph{compatible} with an operator $L\in {\cal L}_{\K[x]}$ if there are $A, B\in\N$, and $\alpha_{k,i}\in \K$ for $k \in \N$, $-A \le i \le B$, such that
\begin{equation}
\label{main}
LP_k = \sum_{i=-A}^B \alpha_{k,i}\;P_{k+i}\ \ {\rm for\ all\ } k \in \N,
\end{equation}
with $P_j = 0$ when $j < 0$. To assert that (\ref{main}) holds for specific $A, B \in \N$, we will say that $\cB$ is \emph{$(A,B)$-compatible} with $L$. \qed 
\end{definition}

%An equivalent characterization of compatible bases is the following:
\begin{proposition}
\label{compat}
A factorial basis $\cB$ of $\K[x]$ is compatible with $L\in {\cal L}_{\K[x]}$ if and only if there are $A, B\in\N$ such that
\begin{enumerate}
\item[\bf C1.] $\deg L P_k \le k + B$ for all $k \ge 0$,
\item[\bf C2.] $P_{k-A} \,|\, L P_k$\ for all $k  \ge A$.
\end{enumerate}
\end{proposition}

\begin{proof}
Necessity of these two conditions  is obvious. For sufficiency, let 
\[
L P_k = \sum_{j=0}^{\deg L P_k} \lambda_{j,k}P_j
\]
be the expansion of $L P_k$ w.r.t.\ $\cB$. By \textbf{C1}, we can replace the upper summation bound by $ k+B$. Rewriting the resulting equation as
\[
LP_k\ - \sum_{j=k-A}^{k+B} \lambda_{j,k} P_j\ =  \sum_{j=0}^{k-A-1} \lambda_{j,k}P_j,
\]
we see  by \textbf{C2} and \textbf{P2} that $P_{k-A}$ divides the left-hand side, while the right-hand side is of degree less than $k-A = \deg P_{k-A}$. Hence both sides vanish, and so
\[
LP_k\ = \sum_{j=k-A}^{k+B} \lambda_{j,k} P_j\ = \sum_{i=-A}^{B} \lambda_{k+i,k} P_{k+i}\ = \sum_{i=-A}^{B} \alpha_{k,i} P_{k+i}
\]
where $ \alpha_{k,i} := \lambda_{k+i,k}$. This proves compatibility of $\cB$ with $L$. \qed
\end{proof}

\begin{example}
\label{basis}
\rm
 \cite[Ex.\,1]{APR} Let
\[
\cP \ =\ \bigg\langle x^k\bigg\rangle_{k=0}^{\infty} 
\]
be the {\em power basis}, and
\[
\cC \ =\ \left\langle {x\choose k}\right\rangle_{k=0}^{\infty} 
\]
the {\em binomial-coefficient basis} of $\K[x]$, respectively. Clearly, both $\cP$ and $\cC$  are factorial bases. Further, let $D$, $E$, $Q$, $X \in {\cal L}_{\K[x]}$  be the \emph{differentiation}, \emph{shift}, \emph{q-shift}, and \emph{multiplication-by-the-independent-variable operators}, respectively, acting on $p \in \K[x]$ by
\[
\begin{array}{lll}
D p(x) &=& p'(x), \\
E p(x) &=& p(x+1), \\
Q p(x) &=& p(qx), \\
X p(x) &=& x p(x)
\end{array}
\]
where $q \in \K^*$ is not a root of unity. Then:
\begin{itemize}
\item $\cP$ is (1,0)-compatible with $D$ (take $\alpha_{k,-1}=k$, $\alpha_{k,0} = 0$),
\item $\cC$ is (1,0)-compatible with $E$ (take $\alpha_{k,-1}= \alpha_{k,0} = 1$),
\item $\cP$ is (0,0)-compatible with $Q$ (take $\alpha_{k,0}=q^k$), 
\item \emph{every} factorial basis is (0,1)-compatible with $X$. \qed
\end{itemize}

\end{example}

Let $\cB = \langle P_k(x) \rangle_{k=0}^{\infty}$ be a factorial basis, and let $\ell_k : \K[x] \rightarrow \K$ for $k \in \N$ be linear functionals such that $\ell_k(P_m)=\delta_{k,m}$ for all $k, m \in \N$ (i.e.,  $\ell_k(p)$ is the coefficient of $P_k$ in the expansion of $p \in \K[x]$ w.r.t.\ $\cB$). Property {\bf P2} implies that $\ell_k(P_j P_m) = 0$ when $k < \max \{j,m\}$. Therefore $\K[x]$ 
naturally embeds into the algebra $\KK$ of \emph{formal polynomial series} of the form 
\begin{equation}
\label{series}
y(x) = \sum_{k=0}^{\infty} c_k P_k(x) \qquad (c_k\in \K),
\end{equation}
with multiplication defined by 
\[
\left(\sum_{k=0}^{\infty} c_k P_k(x)\right)
\left(\sum_{k=0}^{\infty} d_k P_k(x)\right)
\ =\ \sum_{k=0}^{\infty} e_k P_k(x),
\]
\[
e_k\ = \sum_{\max\{i,j\} \le k \le i+j} c_i d_j \,\ell_k(P_i P_j).
\]
Let $\cB$ be compatible with $L \in {\cal L}_{\K[x]}$. Using (\ref{main}), extend $L$ to $\KK$ by defining
\begin{eqnarray}
L \sum_{k=0}^{\infty} c_k P_k(x)  &:=&   \sum_{k=0}^{\infty} c_k L P_k(x) 
\ =\ \sum_{k=0}^{\infty} c_k \sum_{i=-A}^B \alpha_{k,i} P_{k+i}(x) \nonumber \\
&=& \sum_{i=-A}^B \sum_{k=0}^{\infty} \alpha_{k,i} c_{k} P_{k+i}(x) 
\ =\ \sum_{i=-A}^B \sum_{k=i}^{\infty} \alpha_{k-i,i} c_{k-i} P_{k}(x) \nonumber \\
&=& \sum_{k=-A}^{\infty} \sum_{i=-A}^{\min(k,B)} \alpha_{k-i,i} c_{k-i} P_{k}(x) 
\ =\ \sum_{k=0}^{\infty} \sum_{i=-A}^B \alpha_{k-i,i} c_{k-i} P_k(x) \nonumber \\
&=& \sum_{k=0}^{\infty}\left( \sum_{i=-B}^A \alpha_{k+i,-i} c_{k+i}\right) P_k(x) \label{Ly}
\end{eqnarray}
with $A, B, \alpha_{k,i}$ as in (\ref{main}), using $P_k(x) = 0$ for $k < 0$ and $c_{k-i} = 0$ for $i > k$ in the next to last equality. By this definition we have:

\begin{proposition}
\label{RBL}
A formal polynomial series $y(x) = \sum_{k=0}^{\infty} c_k P_k(x) \in\KK$ satisfies $Ly=0$ if and only if its 
coefficient sequence $\langle c_k\rangle_{k\in\Z}$ satisfies the recurrence
\begin{equation}
\label{r1}
\sum_{i=-B}^A \alpha_{k+i,-i} c_{k+i} = 0 \qquad ({\rm for\ all\ } k \ge 0)
\end{equation}
where $c_k=0$ when $k<0$.
\end{proposition}
Hence $L$ naturally induces a recurrence operator
\begin{equation}
\label{RL}
\cR L\ := \sum_{i=-B}^A \alpha_{n+i,-i} E_n^i
\end{equation}
where $E_n$ is the shift operator w.r.t.\ $n$ ($E_n^i c_n = c_{n+i}$ for all $i\in\Z$).

\begin{example}\rm
\label{subst}
Using (\ref{RL}) we read off from Example \ref{basis} that
\begin{eqnarray*}
{\cal R}_{\cP} D &=& (n+1)E_n, \\
{\cal R}_{\cC} E\, &=& E_n+1, \\
{\cal R}_{\cP} Q &=& q^n,
\end{eqnarray*}
while  $x\, x^n = x^{n+1}$ and $x\, {x\choose n} = (n+1){x\choose n+1} + n {x\choose n}$ imply by (\ref{main}) and (\ref{RL}) that
\begin{eqnarray*}
{\cal R}_{\cP} X &=& E_n^{-1}, \\
{\cal R}_{\cC} X\,\! &=& n(E_n^{-1}+1). 
\end{eqnarray*}
\qed
\end{example}

\noindent
\textbf{Notation:}
\begin{itemize}
\item For any factorial basis $\cB$ of $\K[x]$, denote by $\sigma_\cB$ the map $\K[[\cB]] \to \K^{\Z}$ assigning to $y(x) = \sum_{k=0}^\infty c_k P_k(x) \in \K[[\cB]]$ its coefficient sequence $c = \langle c_k\rangle_{k\in\N}$ extended to $\langle c_k\rangle_{k\in\Z}$ by taking $c_k = 0$ whenever $k<0$. We will omit the subscript $\cB$ if it is clear from the context.
\item For any factorial basis $\cB$ of $\K[x]$,  denote by $\cL$ the set of all operators $L
\in {\cal L}_{\K[x]}$ such that  $\cB$ is compatible with $L$.
\item Denote by $\cE$ the $\K$-algebra of recurrence operators of the form $L' = \sum_{i=-S}^R a_i(n) E_n^i$ with $R,S\in\N$ and $a_i: \Z \rightarrow \K$ for $-S \le i \le R$, acting on the $\K$-algebra of all two-way infinite sequences $\K^\Z$. 
\end{itemize}

\begin{proposition}
\label{sigmaLy}
{\rm \cite[Prop.\,1]{APR}}
For any $L \in \cL$ and $y \in \K[[\cB]]$, we have 
\[
\sigma_\cB(Ly) = (\cR L) \sigma_\cB y.
\]
\end{proposition}

\begin{theorem}
\label{iso}
{\rm \cite[Prop.\,2 \& Thm.\,1]{APR}}
$\cL$ is a $\K$-algebra, and the transformation
\[
\cR: \cL \rightarrow \cE,
\]
defined in {\rm (\ref{RL})}, is an isomorphism of $\K$-algebras.
\end{theorem}

By Theorem \ref{iso} and Examples \ref{basis} and \ref{subst}, every linear recurrence operator $L \in \K[x]\langle E\rangle$ is compatible with the basis $\cC \ =\ \left\langle {x\choose n}\right\rangle_{n=0}^{\infty}$, and to compute the associated operator $L' = {\cal R}_{\cC}L \in \cE$, it suffices to apply the substitution
\[
\begin{array}{lll}
E & \mapsto & E_n + 1, \\
x & \mapsto & n(E_n^{-1} + 1)
\end{array}
\]
to all terms of $L$. Every $h \in \ker L'$ then gives rise to a solution $y_n = \sum_{k=0}^\infty \binom{n}{k} h_k$ of $Ly = 0$.

\begin{example}\rm
In this example we list some operators  $L \in \K[n]\langle E\rangle$, their associated operators $L' = {\cal R}_{\cC}L \in \cE$, and some of the elements of their kernels.
\begin{enumerate}
\item $L = E - c$ where $c \in \K^*$: Here $L' = E_n - (c - 1)$, and
\[
y_n\ =\ \sum_{k=0}^\infty \binom{n}{k} (c - 1)^k\ =\ c^n
\]
is indeed a solution of $Ly = 0$.

\item $L = E^2 - 2 E + 1$: Here $L' = E_n^2$, and by {\rm (\ref{r1})}, any $h \in \ker L'$ satisfies $h_{n+2} = 0$ for all $n \ge 0$, or equivalently, $h_n = 0$ for all $n \ge 2$. Hence
\[
y_n\ =\ \sum_{k=0}^\infty \binom{n}{k} h_k\ =\ h_0 + h_1 n
\]
is indeed a solution of $Ly = 0$.

\item $L = E^2 - E - 1$: Here $L' = E_n^2 + E_n - 1$, and any $h \in\ker L'$ is of the form $h_n = (-1)^n (C_1 F_n + C_2 F_{n+1})$ where $C_1, C_2 \in \K$ and $F = \langle 0, 1, 1, 2, \ldots\rangle$ is the sequence of Fibonacci numbers. Hence every $y \in \ker L$ is of the form
\[
y_n\ =\ \sum_{k=0}^\infty \binom{n}{k}  (-1)^k (C_1 F_k + C_2 F_{k+1}).
\]
In particular, by setting $y_0 = F_0$ and $y_1 = F_1$, we discover the identity
\[
F_n\ =\ \sum_{k=0}^n \binom{n}{k}  (-1)^{k+1} F_k.
\]

\item $L = E - (n+1)$: Here $L' = E_n - n - n E_n^{-1}$, and the general solution of $L' h = h_{n+1} - n h_n - n h_{n-1} = 0$ is of the form $h_n = n! (C_1 + C_2 \sum_{k=1}^n (-1)^k/k!)$ where $C_1, C_2 \in \K$. The equation at $n=0$ implies $h_1 = 0$, forcing $C_1 = C_2 =: C$. Hence every $y \in \ker L$ is of the form
\[
y_n\ =\ C \sum_{k=0}^\infty \binom{n}{k} k!\left(1+\sum_{j=1}^k\frac{ (-1)^j}{j!}\right).
\]
In particular, by setting $y_0 = 0! = 1$, we discover the identity
\[
n!\ =\ \sum_{k=0}^n \binom{n}{k} k!\left(1+\sum_{j=1}^k\frac{ (-1)^j}{j!}\right).
\]
or equivalently,
\[
\sum_{k=0}^n \frac{1}{k!} \left(1+\sum_{j=1}^{n-k}\frac{ (-1)^j}{j!}\right)\ =\ 1.
\]

\item $L =  E^3 - ( n^2 + 6 n + 10)E^2 + (n + 2) (2 n + 5) E - (n + 1) (n + 2)$: Unlike in the preceding four cases, the equation $L y = 0$ has no nonzero Liouvillian solutions. Here $L' =  E_n^3 - ( n^2 + 6 n + 7)E_n^2 - (2 n^2 + 8 n + 7) E_n - (n + 1)^2$, and equation $L' h = 0$ has a hypergeometric solution $h_n = n!^2$. So 
\[
y_n\ =\  \sum_{k=0}^\infty \binom{n}{k} k!^2\ =\  \sum_{k=0}^n \binom{n}{k} k!^2
\]
is a non-Liouvillian, definite-sum solution of equation $L y = 0$. \qed

\end{enumerate}

\end{example}

\section{Products of compatible bases}

To be able to use formal polynomial series to find other definite-sum solutions of linear recurrence equations, not just those of the form $\sum_{k=0}^\infty \binom{n}{k} h_k$, we need to construct more factorial bases, compatible with the shift operator $E$. 

\begin{definition}
\label{def:cab} \rm
For $a \in \N \setminus \{0\}$, $b \in \K$, and for all $k \in \N$, let $P_k^{(a,b)}(x) := \binom{ax+b}{k}$. We denote the polynomial basis $\left\langle P_k^{(a,b)}(x) \right\rangle_{k=0}^{\infty}$ by  $\cC_{a,b}$, and call it a \emph{generalized binomial-coefficient basis} of $\K[x]$. \qed
\end{definition}

\begin{proposition}
\label{cab}
Any  generalized binomial-coefficient basis $\cC_{a,b}$ is a factorial basis of $\K[x]$, which is $(a,0)$-compatible with the shift operator $E$.
\end{proposition}

\begin{proof}
Clearly $\deg_x P_k^{(a,b)}(x) = k$ and 
\[
P_{k+1}^{(a,b)}(x) = \frac{ax+b-k}{k+1} P_k^{(a,b)}(x)\ \ {\rm for\ all\ } k \in \N,
\]
so $P_k^{(a,b)}(x) \,\big|\, P_{k+1}^{(a,b)}(x)$, and $\cC_{a,b}$ is factorial. By Chu-Vandermonde's identity, 
\begin{eqnarray*}
E P_k^{(a,b)}(x) &=&  P_k^{(a,b)}(x+1)\ =\ \binom{ax+a+b}{k}\ =\ \,\sum_{i=0}^a \binom{a}{i}\binom{ax+b}{k-i} \\
&=& \sum_{i=-a}^0 \binom{a}{-i}\binom{ax+b}{k+i}\ =\ \sum_{i=-a}^0 \binom{a}{-i}P_{k+i}^{(a,b)}(x),
\end{eqnarray*}
so $\cC_{a,b}$ is $(a,0)$-compatible with $E$
($\alpha_{k,i} = \binom{a}{-i}$ for $i = -a, -a+1, \ldots, 0$). 
\qed
\end{proof}

\begin{definition}
\rm
Let $m \in \N \setminus \{0\}$, and for $i = 1, 2, \ldots, m$, let $\cB_i = \langle P_k^{(i)}(x) \rangle_{k=0}^\infty$ be a basis of $\K[x]$. For all $k \in \N$ and $j \in \{0,1,\ldots,m-1\}$, let
\[
P_{m k + j}^{(\pi)}(x)\ :=\ \prod_{i=1}^j P_{k+1}^{(i)}(x)\cdot \prod_{i=j+1}^m  P_k^{(i)}(x).
\]
Then the sequence $\prod_{i=1}^m \cB_i := \langle P_n^{(\pi)}(x) \rangle_{n=0}^\infty$ is the \emph{product} of $\cB_1, \cB_2, \ldots, \cB_m$. \qed
\end{definition}

\begin{theorem}
\label{prod}
Let $\cB_1, \cB_2, \ldots, \cB_m$ be factorial bases of $\K[x]$, and $L \in {\cal L}_{\K[x]}$.
\begin{enumerate}
\item $\prod_{i=1}^m \cB_i$ is a factorial basis of $\K[x]$.
\item Let $L$ be an endomorphism of $\K[x]$, and let each $\cB_i$ be $(A_i, B_i)$-compatible with $L$.  Write $A = \max_{1 \le i \le m} A_i$ and $B = \min_{1 \le i \le m} B_i$. Then $\prod_{i=1}^m \cB_i$ is $(m A, B)$-compatible with $L$.
\end{enumerate}
\end{theorem}

\begin{proof}
\begin{enumerate}
\item Clearly $\deg P_{m k + j}^{(\pi)} =  j(k+1) + (m-j)k = mk+j$.

If $n = mk+j$ with $0 \le j \le m-2$, then $n+1 = mk+(j+1)$ and
\[
\frac{P_{n+1}^{(\pi)}}{P_n^{(\pi)}}\ =\ \frac{\prod_{i=1}^{j+1} P_{k+1}^{(i)}}{\prod_{i=1}^j P_{k+1}^{(i)}} \cdot \frac{\prod_{i=j+2}^m  P_k^{(i)}}{\prod_{i=j+1}^m  P_k^{(i)}}\ =\ \frac{ P_{k+1}^{(j+1)}}{P_k^{(j+1)}} 
\ \in\ \K[x]
\]
as $\cB_{j+1}$ is factorial. If $n = mk+(m-1)$, then $n+1 = m(k+1) + 0$ and
\[
\frac{P_{n+1}^{(\pi)}}{P_n^{(\pi)}}\ =\ \frac{1}{\prod_{i=1}^{m-1} P_{k+1}^{(i)}} \cdot \frac{\prod_{i=1}^m  P_{k+1}^{(i)}}{P_k^{(m)}}\ =\ \frac{ P_{k+1}^{(m)}}{P_k^{(m)}} 
\ \in\ \K[x]
\]
because $\cB_{m}$ is factorial. Hence $\prod_{i=1}^m \cB_i$ is factorial as well.

\item Let $p \in \K[x]$ be arbitrary. For $i = 1,2,\ldots,m$, let $p = \sum_{k=0}^{\deg p} c_k^{(i)} P_k^{(i)}$ be the expansion of $p$ w.r.t.\ $\cB_i$. Then  $L p = \sum_{k=0}^{\deg p} c_k^{(i)} L P_k^{(i)}$, and by condition \textbf{C1} of Proposition \ref{compat},
\[
\deg L p\ \le \max_{0 \le k \le \deg p} \deg L P_k^{(i)} \le \max_{0 \le k \le \deg p} (k + B_i)\ =\ \deg p + B_i.
\]
Since this holds for all $i$, we have $\deg L p\ \le\ \deg p + B$ for all $p \in \K[x]$. In particular, $\deg L P_k^{(\pi)}\ \le\ k + B$, so $\prod_{i=1}^m \cB_i$ satisfies \textbf{C1}.

Condition \textbf{C2} of Proposition \ref{compat} and our definition of $A$ imply that $P_{k+1-A}^{(i)} \,\big|\, L P_{k+1}^{(i)}$ and $P_{k-A}^{(i)} \,\big|\, L P_{k}^{(i)}$ for all $k \ge A$ and $i \in \{1,2\ldots,m\}$,~so
\[
 P_{m (k-A) + j}^{(\pi)}\ =\ \prod_{i=1}^j P_{k+1-A}^{(i)}\cdot \prod_{i=j+1}^m  P_{k-A}^{(i)} \ \ \bigg|\ \ \prod_{i=1}^j L P_{k+1}^{(i)}\cdot \prod_{i=j+1}^m  L P_k^{(i)},
\]
or equivalently, since $L$ is an endomorphism of the ring $\K[x]$,
\[
 P_{m (k-A) + j}^{(\pi)} \ \bigg|\  L \left(\prod_{i=1}^j P_{k+1}^{(i)}\cdot \prod_{i=j+1}^m  P_k^{(i)}\right) \ =\ L  P_{m k + j}^{(\pi)}. 
\]
For $n = mk+j \ge m A$, this turns into $P_{n - m A}^{(\pi)}\ \big|\,\ L  P_{n}^{(\pi)}$, so $\prod_{i=1}^m \cB_i$ satisfies \textbf{C2} as well. By Proposition \ref{compat}, this proves the claim.  \qed

\end{enumerate}
\end{proof}

\begin{definition}
\label{def:caabb} \rm
Let $m \in \N \setminus \{0\}$, and let $\mathbf{a} = (a_1,a_2,\ldots,a_m)$, $\mathbf{b} = (b_1,b_2,\ldots,b_m)$ where $a_i \in \N \setminus \{0\}$, $b_i \in \K$ for $i = 1,2,\ldots,m$. We denote the product of generalized binomial-coefficient bases $\prod_{i=1}^m \cC_{a_i, b_i}$ by $\cC_{\mathbf{a},\mathbf{b}}$, and call it a \emph{product binomial-coefficient basis} of $\K[x]$. \qed
\end{definition}

\begin{corollary}
\label{cor:caabb} \rm
Any product binomial-coefficient basis $\cC_{\mathbf{a},\mathbf{b}}$  is a factorial basis of $\K[x]$ which is $(m A, 0)$-compatible with $E$, where $A = \max_{1\le i\le m}  a_i$.
\end{corollary}

\begin{proof}
\noindent
Use Theorem \ref{prod} and Proposition \ref{cab}. \qed
\end{proof}

\section{Expansions in binomial-coefficient bases}

By Corollary \ref{cor:caabb}, we now have at our disposal a rich family of factorial bases  $\cC_{\mathbf{a},\mathbf{b}}$, compatible with any operator $L \in \K[x]\langle E\rangle$, which, given $L$ and $F(n,k)$ of the form (\ref{binom}), can be used to find solutions of the form $y_n = \sum_{k=0}^\infty F(n,k)\, h_k$ of $L y = 0$. To this end, we need to compute expansions of $E P_n(x)$ and $X P_n(x)$ in the basis  $\cC_{\mathbf{a},\mathbf{b}}=  \langle P_n(x)\rangle_{n=0}^\infty$.

\begin{example}
\label{xk2E}
\rm
For the simplest nontrivial example, take $F(n,k) = \binom{n}{k}^2$. The polynomial basis to be used here is $\cC_{(1,1),(0,0)} = \langle P_n(x)\rangle_{n=0}^\infty$ where for all $k \in \N$,
\[
P_{2k}(x) =  \binom{x}{k}^2, \quad P_{2k+1}(x) =  \binom{x}{k+1}  \binom{x}{k}.
\]
According to Corollary \ref{cor:caabb}, $\cC_{(1,1),(0,0)}$ is a factorial basis of $\K[x]$ with $m = 2$ and $A = \max\{1,1\} = 1$, so it is $(2,0)$-compatible with $E$. In particular, this means that $P_{2k}(x+1)$ can be expressed as a linear combination of $P_{2k}(x)$, $P_{2k-1}(x)$ and $P_{2k-2}(x)$, and $P_{2k+1}(x+1)$ as a linear combination of $P_{2k+1}(x)$, $P_{2k}(x)$ and $P_{2k-1}(x)$, with coefficients depending on $k$. In the case of $P_{2k}(x+1)$ this is just an application of Pascal's rule:
\begin{eqnarray}
P_{2k}(x+1) &=&  \binom{x+1}{k}^2\ =\ \left[\binom{x}{k} +  \binom{x}{k-1}\right]^2 \nonumber\\
&=&  \binom{x}{k}^2 +  2 \binom{x}{k}\binom{x}{k-1} + \binom{x}{k-1}^2 \nonumber\\[3pt]
&=& P_{2k}(x) + 2  P_{2k-1}(x) +  P_{2k-2}(x). \label{p2k}
\end{eqnarray}
 In the case of $P_{2k+1}(x+1)$, we can use the method of undetermined coefficients. Dividing both sides of
\begin{eqnarray*}
P_{2k+1}(x+1) &=&  u(k) P_{2k+1}(x) + v(k)  P_{2k}(x) +  w(k) P_{2k-1}(x), {\rm\ \ or} \\
\binom{x+1}{k+1} \binom{x+1}{k} &=& u(k) \binom{x}{k+1}\binom{x}{k} +  v(k) \binom{x}{k}^2 + w(k) \binom{x}{k}\binom{x}{k-1}
\end{eqnarray*}
where $u(k), v(k), w(k)$ are undetermined functions of $k$, by $\binom{x}{k}\binom{x}{k-1}$, yields
\begin{equation}
\label{undetcoef}
\frac{(x+1)^2}{k(k+1)}\ =\ u(k) \frac{(x-k+1)(x-k)}{k(k+1)} + v(k) \frac{x-k+1}{k} + w(k),
\end{equation}
which is an equality of two quadratic polynomials from $\K(k)[x]$. Plugging in the values $x = -1, k, k-1$, we obtain a triangular system of linear equations
\[
\begin{array}{rcc}
u(k) - v(k) + w(k) &=& 0 \\[2pt]
\frac{1}{k} v(k) + w(k) &=& \frac{k+1}{k} \\[3pt]
w(k) &=& \frac{k}{k+1}
\end{array}
\]
whose solution is $u(k) = 1$ (as expected), $v(k) = \frac{2k+1}{k+1}$,  $w(k) = \frac{k}{k+1}$, and so
\begin{equation}
\label{p2k1}
P_{2k+1}(x+1)\ =\  P_{2k+1}(x) + \frac{2k+1}{k+1}  P_{2k}(x) +  \frac{k}{k+1} P_{2k-1}(x).
\end{equation}

For the expansion of $X P_n(x)$, recall that every factorial basis is $(0,1)$-compatible with $X$. Indeed, as $x \binom{x}{k} =  k\binom{x}{k} + (k+1)\binom{x}{k+1}$, we have
\begin{align*}
x P_{2k}(x)\ &=\ \binom{x}{k} \left[x \binom{x}{k}\right]\ =\ (k+1) P_{2k+1}(x) + k P_{2k}(x), \\
x P_{2k+1}(x)\ &=\ \binom{x}{k+1} \left[x \binom{x}{k}\right]\ =\ (k+1) P_{2k+2}(x) + k P_{2k+1}(x).
\end{align*}
\qed
\end{example}

\label{expansion}
In the general case $F(n,k) = \prod_{i=1}^m \binom{a_i n + b_i}{k}$, we use the basis $\cC_{\mathbf{a},\mathbf{b}} = \langle P_n^{(\pi)}(x)\rangle_{n=0}^\infty$ which, by Corollary \ref{cor:caabb}, is $(m A, 0)$-compatible with $E$ where $A = \max_{1\le i\le m} a_i$. In order to compute $\alpha_{k,j,i} \in \K(k)$ such that
\[
P_{mk+j}^{(\pi)}(x+1)\ =\ \sum_{i=0}^{mA} \alpha_{k,j,i} P_{mk+j-i}^{(\pi)}(x)
\]
for all $k \in \N$ and $j \in \{0,1,\ldots,m-1\}$, we divide both sides of this equation by $P_{mk+j-mA}^{(\pi)}(x)$ which turns it into an equality of two polynomials of degree $mA$ from $\K(k)[x]$. From this equality a system of $mA+1$ linear algebraic equations for the $mA+1$ undetermined coefficients $\alpha_{k,j,i}$, $i = 0, 1, \ldots, mA$, can be obtained by equating the coefficients of like powers of $x$ on both sides, or (as in Example \ref{xk2E}) by substituting $mA+1$ distinct values from $\K(k)$ for $x$ in this equality. Note that for each $j \in \{0,1,\ldots,m-1\}$, this system is uniquely solvable since $\cC_{\mathbf{a},\mathbf{b}}$ is a basis of $\K[x]$, that the $\alpha_{k,j,i}$ will be rational functions of $k$, and that, as the shift operator preserves leading coefficients and degrees of polynomials, $\alpha_{k,j,0} = 1$.

\begin{example}
\label{expandE}
\rm
Here we give some additional examples of expansions of a shifted basis element in the basis $\cC_{\mathbf{a},\mathbf{b}}$, computed by the procedure just described.

\medskip
\noindent
(a)\ \ $\cC_{\mathbf{a},\mathbf{b}} = \cC_{(2,3),(0,0)} = \langle P_n(x)\rangle_{n=0}^\infty$ where for all $k \in \N$,
\[
P_{2k}(x) =  \binom{2x}{k} \binom{3x}{k}, \quad P_{2k+1}(x) =  \binom{2x}{k+1} \binom{3x}{k}
\]
Here $m=2$, $A = 3$, $mA = 6$ and
\begin{eqnarray*}
\lefteqn{P_{2 k}(x + 1)\ =\ P_{2 k}(x)\ +\ 6 P_{2 k-1}(x)}\\
&+& \frac{3 (7 k-3)}{2 k} P_{2 k-2}(x)\ +\ \frac{131 k-64}{12 k} P_{2 k-3}(x) \\
&+& \frac{211 k^2-374 k+120}{36 (k-1) k} P_{2 k-4}(x)\ +\ \frac{2 (2 k-3)}{9(k-1)} P_{2 k-5}(x), \\
\lefteqn{P_{2 k+1}(x + 1)\ =\ P_{2 k+1}(x)\ +\ \frac{2 (2 k+1)}{k+1} P_{2 k}(x)\ +\ \frac{17 k+7}{2 (k+1)} P_{2 k-1}(x)}\\
&+& \frac{131 k^2-6 k-17}{18 k (k+1)} P_{2 k-2}(x) 
\ +\ \frac{2 \left(10 k^2-6 k-1\right)}{9 k (k+1)} P_{2 k-3}(x)\\
&+& \frac{4 (k-2) (2 k-3)}{27 (k-1) (k+1)} P_{2 k-4}(x) -\frac{2 k (2 k-3)}{27 (k-1)(k+1)}  P_{2 k-5}(x).
\end{eqnarray*}

\bigskip
\noindent
(b)\ \ $\cC_{\mathbf{a},\mathbf{b}} = \cC_{(2,3),(-1,4)} = \langle P_n(x)\rangle_{n=0}^\infty$ where for all $k \in \N$,
\[
P_{2k}(x) =  \binom{2x-1}{k} \binom{3x+4}{k}, \quad P_{2k+1}(x) =  \binom{2x-1}{k+1} \binom{3x+4}{k}
\]
Here $m=2$, $A = 3$, $mA = 6$ and
\begin{eqnarray*}
\lefteqn{P_{2 k}(x + 1)\ =\ P_{2 k}(x)\ +\ 6 P_{2 k-1}(x)}\\
&+& \frac{21 k+13}{2 k} P_{2 k-2}(x)\ +\ \frac{131 k-97}{12 k} P_{2 k-3}(x) \\
&+& \frac{211 k^2+330 k+791}{36 (k-1) k} P_{2 k-4}(x)\ +\ \frac{2 (k-7) (2 k-11)}{9 (k-1) k} P_{2 k-5}(x), \\
\lefteqn{P_{2 k+1}(x + 1)\ =\ P_{2 k+1}(x)\ +\ \frac{2 (2 k+1)}{k+1} P_{2 k}(x)\ +\ \frac{17 k-15}{2 (k+1)} P_{2 k-1}(x)}\\
&+& \frac{131 k^2-39 k+214}{18 k (k+1)} P_{2 k-2}(x) 
\ +\ \frac{4 \left(5 k^2-47 k+104\right)}{9 k (k+1)} P_{2 k-3}(x)\\
&+& \frac{4 (k-7) (k-2) (2 k-11)}{27 (k-1) k (k+1)} P_{2 k-4}(x) -\frac{2 (k-7) (k+11) (2 k-11)}{27 (k-1) k (k+1)}  P_{2 k-5}(x).
\end{eqnarray*}

\bigskip
\noindent
(c)\ \ $\cC_{\mathbf{a},\mathbf{b}} = \cC_{(4,4),(0,0)} = \langle P_n(x)\rangle_{n=0}^\infty$ where for all $k \in \N$,
\[
P_{2k}(x) =  \binom{4x}{k}^2, \quad P_{2k+1}(x) =  \binom{4x}{k+1} \binom{4x}{k}
\]
Here $m=2$, $A = 4$, $mA = 8$ and
\begin{eqnarray*}
\lefteqn{P_{2 k}(x + 1)\ =\ P_{2 k}(x)\ +\ 8 P_{2 k-1}(x)\ +\ \frac{4 (7 k-3)}{k} P_{2 k-2}(x)}\\
&+& \frac{28 (2 k-1)}{k}  P_{2 k-3}(x)\ +\ \frac{2 \left(35 k^2-63 k+22\right)}{(k-1) k} P_{2 k-4}(x)\\
&+& \frac{8 \left(7 k^2-14 k+5\right)}{(k-1) k} P_{2 k-5}(x)
\ +\ \frac{4 \left(7 k^3-28 k^2+32 k-9\right)}{(k-2) (k-1) k} P_{2 k-6}(x)\\
&+& \frac{4 (2 k-3) \left(k^2-3 k+1\right)}{(k-2) (k-1) k} P_{2 k-7}(x)\ +\ P_{2 k-8}(x), \\
\lefteqn{P_{2 k+1}(x + 1)\ =\ P_{2 k+1}(x)\ +\ \frac{4 (2 k+1)}{k+1} P_{2 k}(x)\ +\ \frac{4 (7 k+3)}{k+1} P_{2 k-1}(x)}\\
&+& \frac{8 \left(7 k^2-1\right)}{k (k+1)} P_{2 k-2}(x) 
\ +\ \frac{2 \left(35 k^2-7 k-6\right)}{k (k+1)} P_{2 k-3}(x)\\
&+& \frac{4 (2 k-1) \left(7 k^2-7 k-2\right)}{(k-1) k (k+1)} P_{2 k-4}(x)
\ +\ \frac{4 \left(7 k^3-14 k^2+4 k+1\right)}{(k-1) k (k+1)}  P_{2 k-5}(x)\\
&+& \frac{8 (k-1)}{k+1} P_{2 k-6}(x)\ +\ \frac{k-3}{k+1}  P_{2 k-7}(x).
\end{eqnarray*}
\qed
\end{example}
As every factorial basis, $\cC_{\mathbf{a},\mathbf{b}} = \langle P_n^{(\pi)}(x)\rangle_{n=0}^\infty$ is $(0,1)$-compatible with $X$:
%so we have only two coefficients to determine in the expansion of $x P_n^{(\pi)}(x)$.

\begin{proposition}
\label{x}
For $k \in \N$ and $j \in \{0,1,\ldots,m-1\}$, let
\begin{equation}
\label{Pmkj}
P_{m k + j}^{(\pi)}(x)\ :=\ \prod_{i=1}^j \binom{a_i x+b_i}{k+1}\cdot \prod_{i=j+1}^m \binom{a_i x+b_i}{k}.
\end{equation}
Then
\[
x P_{m k + j}^{(\pi)}(x)\ =\ \frac{k+1}{a_{j+1}} P_{m k + j + 1}^{(\pi)}(x) + \frac{k-b_{j+1}}{a_{j+1}} P_{m k + j}^{(\pi)}(x).
\]
\end{proposition}

\begin{proof}
\[
\frac{k+1}{a_{j+1}} P_{m k + j + 1}^{(\pi)}(x) + \frac{k-b_{j+1}}{a_{j+1}} P_{m k + j}^{(\pi)}(x)\ =\ P_{m k + j}^{(\pi)}(x) \cdot f(x)
\]
where
\begin{eqnarray*}
f(x) &=& \frac{k+1}{a_{j+1}} \cdot \frac{P_{m k + j + 1}^{(\pi)}(x)}{P_{m k + j}^{(\pi)}(x)} + \frac{k-b_{j+1}}{a_{j+1}}\\
&=& \frac{k+1}{a_{j+1}} \cdot \frac{\binom{a_{j+1} x+b_{j+1}}{k+1}}{\binom{a_{j+1} x+b_{j+1}}{k}} + \frac{k-b_{j+1}}{a_{j+1}}\\ 
&=& \frac{k+1}{a_{j+1}} \cdot \frac{a_{j+1} x + b_{j+1} - k}{k+1} + \frac{k-b_{j+1}}{a_{j+1}}\ =\ x. 
\end{eqnarray*}
\qed
\end{proof}

\section{Finding definite-sum solutions of recurrences}

Now we can use Proposition \ref{RBL} and the associated operator $\cR L$ where $\cB = \cC_{\mathbf{a},\mathbf{b}}$ to find $h$ such that $y_n = \sum_{k=0}^\infty F(n,k) h_k$, with $F$ as in (\ref{binom}), satisfies $Ly = 0$. Notice however that for $m > 1$, the coefficients $\alpha_{k,i}$ expressing the actions of $E$ resp.\ $X$ on $\cB$ are not rational functions of $k$ anymore, but conditional expressions evaluating to $m$ generally distinct rational functions, depending on the residue class of $k \bmod m$ (cf.\ Example \ref{expandE} and Proposition \ref{x}). So the coefficients of $\cR L$, obtained by composing the operators $\cR E$ and $\cR X$ repeatedly, will contain quite complicated conditional expressions. In addition, $\ord \cR L$ may exceed $\ord L$ by a factor of $mA$ which can be exponential in input size. Finally, only those $h \in \ker \cR L$ that satisfy $h_k = 0$ whenever $k \not\equiv 0 \pmod m$ (i.e., those that are interlacings of an arbitrary sequence with $m-1$ consecutive 0 sequences) give rise to elements of $\ker L$ that have the desired form. 

To overcome these inconveniences, for a fixed $m > 1$ we do not attempt to compute $\cR L$ directly but represent it by a matrix $\left[\cR L\right] = \left[L_{r,j}\right]_{r,j=0}^{m-1}$  of operators where $L_{r,j} \in \cE$ expresses the contribution of the $j$-th $m$-section  $s_j^m \sigma_\cB y$ of the coefficient sequence of $y$ to the $r$-th $m$-section $s_r^m \sigma_\cB (L y)$ of the coefficient sequence of $Ly$.

\begin{proposition}
\label{mainprop}
Let $L \in {\cal L}_{\K[x]}$,  $\cB = \langle P_n(x)\rangle_{n=0}^\infty$ {\rm (}a factorial basis of $\K[x]${\rm )}, $m \in \N \setminus \{0\}$, and  $A, B \in \N$ be such that for all $k \in \N$ and $j \in \{0,1,\ldots,m-1\}$,
\begin{equation}
\label{LPskj}
L P_{mk+j}(x) = \sum_{i=-A}^B \alpha_{k,j,i} P_{mk+j+i}(x).
\end{equation}
Furthermore, for all $r,j \in \{0,1,\ldots,m-1\}$ define
\begin{eqnarray}
\label{Lrjdef}
\ L_{r,j}\ := \!\!\!\!\!\sum_{\,-A \le i \le B \ \ \atop \, i+j \equiv r \!\!\!\!\!\pmod{\!m}}\!\!\!\!\!\alpha_{k+\frac{r-i-j}{m},j,i} E_k^{\frac{r-i-j}{m}} \in\ \cE 
\end{eqnarray}
{\rm (}to keep notation simple, we do not make the dependence of $ L_{r,j}$ on $m$ explicit{\rm )}.
Then for every $y \in \K[[\cB]]$  and $r \in \{0,1,\ldots,m-1\}$,
\begin{equation}
\label{Lrj}
 s_r^m \sigma_\cB (L y)\ =\ \sum_{j=0}^{m-1} L_{r,j}\,  s_j^m \sigma_\cB y.
\end{equation}
\end{proposition}

\begin{proof}
Write $y(x) = \sum_{n=0}^\infty c_n P_n(x)$ as the sum of its $m$-sections
\begin{equation}
\label{y}
y(x) = \sum_{j=0}^{m-1} \sum_{k=0}^\infty c_{mk+j} P_{mk+j}(x) = \sum_{j=0}^{m-1} \sum_{k=0}^\infty (s_j^mc)_k P_{mk+j}(x).
\end{equation}
Then
\begin{align}
L y(x) &=\  \sum_{j=0}^{m-1} \sum_{k=0}^\infty (s_j^mc)_k L P_{mk+j}(x) = \sum_{j=0}^{m-1} \sum_{k=0}^\infty  (s_j^mc)_k \sum_{i=-A}^B \alpha_{k,j,i} P_{mk+j+i}(x) \label{Ly(x)} \\
&=\ \sum_{j=0}^{m-1}\sum_{r=0}^{m-1} \sum_{\,-A \le i \le B \ \ \atop \,i+j \equiv r \!\!\!\!\!\pmod{\!m}} \sum_{k=0}^\infty \alpha_{k,j,i}\, \left(s_j^mc\right)_k\,  P_{mk+i+j}(x) \label{residue} \\
&= \sum_{r,j=0}^{m-1} \sum_{\,-A \le i \le B \ \ \atop \, i+j \equiv r \!\!\!\!\!\pmod{\!m}} \sum_{k=\frac{i+j-r}{m}}^\infty \alpha_{k+\frac{r-i-j}{m},j,i}\, \left(s_j^mc\right)_{k+\frac{r-i-j}{m}}\,  P_{mk+r}(x)\label{subsk} \\
&= \sum_{r,j=0}^{m-1} \sum_{\,-A \le i \le B \ \ \atop \, i+j \equiv r \!\!\!\!\!\pmod{\!m}} \sum_{k=0}^\infty \alpha_{k+\frac{r-i-j}{m},j,i}\left(E_k^{\frac{r-i-j}{m}}s_j^mc\right)_k  P_{mk+r}(x) \label{zero}\\
&=\ \sum_{r=0}^{m-1} \sum_{k=0}^\infty \bigg(\sum_{j=0}^{m-1} L_{r,j}\, s_j^mc\bigg)_k\,  P_{mk+r}(x) \label{final}
%\ =\ \sum_{r=0}^{m-1} \sum_{k=0}^\infty \left(L_{r}^{(s)} c\right)_k\,  P_{mk+r}(x) 
\end{align}
where in (\ref{Ly(x)}) we used  (\ref{LPskj}), in (\ref{residue}) we reordered summation on $i$ with respect to the residue class of $i+j \bmod m$, (\ref{subsk}) was obtained by replacing $k$ with $k - \frac{i+j-r}{m}$, (\ref{zero}) by noting that 
\begin{eqnarray*}
k < 0\ \ \Longrightarrow\ \ mk+r < 0\ \ \Longrightarrow\ \ P_{mk+r}=0, \\
k < \frac{i+j-r}{m}\ \ \Longrightarrow\ \ \left(s_j^mc\right)_{k+\frac{r-i-j}{m}} = 0,
\end{eqnarray*}
and (\ref{final}) by using (\ref{Lrjdef}). Now the equality of $Ly(x)$ and the series in (\ref{final}) can be restated as (\ref{Lrj}). \qed
\end{proof}

\begin{corollary}
\label{cor}
Under the assumptions of Proposition \ref{mainprop}, 
\[
L y = 0\quad \Longleftrightarrow\quad  \forall r \in \{0,1,\ldots,m-1\}\!: \sum_{j=0}^{m-1} L_{r,j}\,  s_j^m \sigma_\cB y\ =\ 0.
\]
\end{corollary}

Note that for $m=1$, Corollary \ref{cor} and  Proposition \ref{mainprop} turn into Proposition \ref{RBL} and Proposition \ref{sigmaLy}, respectively (with $\alpha_{k,i} = \alpha_{k,0,i}$ and $\cR L = L_{0,0}$).

\begin{proposition}
\label{xrjcab}
Let $\cB = \cC_{\mathbf{a},\mathbf{b}}$. Then for all $r,j \in \{0,1,\ldots,m-1\}$,
\[
X_{r,j}\ =\ [r = j]\,\frac{k-b_{j+1}}{a_{j+1}} +
[r = 0 \land j = m - 1]\,\frac{k}{a_{j+1}}\,E_k^{-1} +
[r = j+1]\,\frac{k+1}{a_{j+1}}
\]
where 
\[
[\varphi]\ =\ \left\{
\begin{array}{ll}
1, & {\rm\ if\ } \varphi {\rm\ is\ true}, \\
0, & {\rm\ otherwise}
\end{array}
\right.
\]
is Iverson bracket.
\end{proposition}

\begin{proof}
From Proposition \ref{x} we read off that in this case
\begin{align}
\alpha_{k,j,0} &\ =\ \frac{k-b_{j+1}}{a_{j+1}}, \label{kj0} \\
\alpha_{k,j,1} &\ =\ \frac{k+1}{a_{j+1}}. \label{kj1}
\end{align}
From (\ref{Lrjdef}) with $A = 0$, $B = 1$ it follows that
\begin{align}
X_{r,j} &\ =\ [j\equiv r \!\!\!\!\!\pmod{m}]\,\alpha_{k+\frac{r-j}{m},j,0} E_k^{\frac{r-j}{m}}\nonumber  \\
&\ +\ [j\equiv r-1 \!\!\!\!\!\pmod{m}]\,\alpha_{k+\frac{r-j-1}{m},j,1} E_k^{\frac{r-j-1}{m}}. \label{xrjm}
\end{align}
Combining (\ref{kj0}) -- (\ref{xrjm}) with $0 \le r,j \le m-1$ yields the assertion. \qed
\end{proof}

\bigskip
\noindent
\textbf{Notation:} $\left[\cR L\right] := \left[L_{r,j}\right]_{r,j=0}^{m-1} \in M_m(\cE)$ where $L_{r,j}$ is as given in (\ref{Lrjdef}).

\begin{proposition}
Let $L^{(1)}, L^{(2)} \in \cL$. Then
\[
\left[\cR \left(L^{(1)} L^{(2)}\right)\right]\ =\ \left[\cR L^{(1)}\right]\!\!\left[\cR L^{(2)}\right].
\]
\end{proposition}

\begin{proof}
Write $L = L^{(1)} L^{(2)}$. By (\ref{Lrj}),
\begin{eqnarray*}
s_t^m \sigma (L y) \ =\ \sum_{j=0}^{m-1} L_{t,j}\,  s_j^m \sigma y\ =\  \sum_{j=0}^{m-1}\left[\cR L\right]_{t,j} s_j^m \sigma y.
\end{eqnarray*}
On the other hand, by (\ref{Lrj}) applied to $L^{(1)}$ and $L^{(2)}$,
\begin{eqnarray*}
s_t^m \sigma (L y) &=&s_t^m \sigma (L^{(1)}L^{(2)} y)
\ =\ \sum_{r=0}^{m-1} L_{t,r}^{(1)}\, s_r^m \sigma (L^{(2)}y) \\
 &=& \sum_{r=0}^{m-1} L_{t,r}^{(1)}\sum_{j=0}^{m-1} L_{r,j}^{(2)}\,  s_j^m \sigma y 
\ = \sum_{j=0}^{m-1} \left(\sum_{r=0}^{m-1} L_{t,r}^{(1)} L_{r,j}^{(2)}\right) s_j^m \sigma y \\
&=& \sum_{j=0}^{m-1} \left(\left[\cR L^{(1)}\right]\!\!\left[\cR L^{(2)}\right]\right)_{t,j} s_j^m \sigma y.
\end{eqnarray*}
Hence
\begin{eqnarray*}
\sum_{j=0}^{m-1}\left[\cR L\right]_{t,j} s_j^m \sigma y\ = \sum_{j=0}^{m-1} \left(\left[\cR L^{(1)}\right]\!\!\left[\cR L^{(2)}\right]\right)_{t,j} s_j^m \sigma y
\end{eqnarray*}
for any $y \in \K[[\cB]]$, which implies the claim. \qed
\end{proof} 

It follows that to compute $[\cR L]$ for an arbitrary operator $L \in \K[x]\langle E\rangle$, it suffices to apply the substitution
\begin{equation}
\label{substitute}
\begin{array}{llc}
E & \mapsto & [\cR E], \\
x & \mapsto & [\cR X], \\
1 & \mapsto & I_m
\end{array}
\end{equation}
where $I_m$ is the $m \times m$ identity matrix, to all terms of $L$. Since we are interested in finding $y \in \ker L$ of the form 
\[
y(x)\ =\ \sum_{k=0}^\infty h_k \prod_{i=1}^m \binom{a_i x+b_i}{k}\ =\ \sum_{k=0}^\infty h_k P_{mk}(x),
\]
we have $s_0^m \sigma y = h$ and $s_j^m \sigma y = 0$ for all $j \ne 0$. For such $y$, Corollary \ref{cor} implies
\begin{eqnarray}
L y = 0 &\Longleftrightarrow&  \forall r \in \{0,1,\ldots,m-1\}\!: L_{r,0}\, h\ =\ 0 \nonumber \\
&\Longleftrightarrow&  {\rm gcrd}(L_{0,0}, L_{1,0},\ldots, L_{m-1,0})\, h\ =\ 0. \label{Lr0}
\end{eqnarray}
So any nonzero element of the first column of $[\cR L] = \left[L_{r,j}\right]_{r,j=0}^{m-1}$ may serve as a nontrivial annihilator $L'$ of $h$, and taking their greatest common right divisor might yield $L'$ of lower order. The fact that we only need the first column $[\cR L]e^{(1)}$ of $[\cR L]$ (where $e^{(1)} = (1,0,\ldots,0)^{\rm T}$ is the first standard basis vector of length $m$) can also be used to advantage in its computation.

In summary, we have the following algorithm:

\smallskip
\begin{center}
\large
\textbf{Algorithm} \textsc{DefiniteSumSolutions}
\end{center}

\smallskip
\textsc{Input:} \verb+   + $L \in \K[x]\langle E\rangle$, $m \in \N\setminus\{0\}$,

\verb+          + $a_1,a_2,\ldots,a_m \in  \N\setminus\{0\}$,  $b_1,b_2,\ldots,b_m \in \Z$

\medskip
\textsc{Output:} \verb+ + $L' \in \K(k)\langle E_k\rangle$ such that 
\[
L\left(\sum_{k=0}^\infty \prod_{i=1}^m \binom{a_i x + b_i}{k} h_k\right) = 0
\]

\verb+          + if and only if\ $L' h = 0$

\vskip 1pc
\begin{enumerate}
\item $A := \max_{1 \le i \le m} a_i$.
\item For $j = 0,1, \ldots,m-1$ let
\[
P_{mk+j}(x) \ :=\ \prod_{i=1}^j \binom{a_i x + b_i}{k+1} \cdot \prod_{i=j+1}^m \binom{a_i x + b_i}{k}.
\]
For $j = 0,1,\ldots,m-1$ compute $\alpha_{k,j,i} \in \K(k)$ such that
\[
P_{mk+j}(x+1)\ =\ \sum_{i = -m A}^0 \alpha_{k,j,i} P_{mk+j+i}(x)
\]
as explained below Example \ref{xk2E} on p.\ \pageref{expansion}.
\item For $r,j = 0,1,\ldots,m-1$ let
\begin{align*}
E_{r,j} &:= \!\!\!\!\!\sum_{\,-mA \le i \le 0 \ \ \atop \, i+j \equiv r \!\!\!\!\!\pmod{\!m}}\!\!\!\!\!\alpha_{k+\frac{r-i-j}{m},j,i} E_k^{\frac{r-i-j}{m}},\\
X_{r,j} &:= [r = j]\frac{k-b_{j+1}}{a_{j+1}} +
[r = 0 \land j = m - 1]\frac{k}{a_{j+1}}\,E_k^{-1} +
[r = j+1]\frac{k+1}{a_{j+1}}.
\end{align*}
Let $[\cR E] =\left[E_{r,j}\right]_{r,j=0}^{m-1}$,\ \ $[\cR X] =\left[X_{r,j}\right]_{r,j=0}^{m-1}$.

\item Let $[\cR L] =\left[L_{r,j}\right]_{r,j=0}^{m-1}$ be the matrix of operators obtained by applying substitution (\ref{substitute}) to $L$.

\item Return $L' := {\rm gcrd}(L_{0,0}, L_{1,0},\ldots, L_{m-1,0})$.
\end{enumerate}
\begin{remark}
In step 4 it suffices to compute $[\cR L]e^{(1)}$, the first column of  $[\cR L]$. We do so by proceeding from right to left through the expression obtained from $L$ by substitution (\ref{substitute}), multiplying a matrix with a vector at each point.
\end{remark}

\begin{example}\rm
Consider again the basis $\cB = \cC_{(1,1),(0,0)}$ from Example \ref{xk2E}.  Here $m=2$, $a_1 = a_2 = 1$, $b_1 = b_2 = 0$. To compute $[\cR E]$, comparing (\ref{LPskj}) with (\ref{p2k}) and (\ref{p2k1}) yields
\[
\begin{array}{ccccccc}
\alpha_{k,0,0} & = & 1 & \qquad & \alpha_{k,1,0} & = & 1 \\[2pt]
\alpha_{k,0,-1} & = & 2 & \qquad & \alpha_{k,1,-1} & = & \frac{2k+1}{k+1} \\[3pt]
\alpha_{k,0,-2} & = & 1 & \qquad & \alpha_{k,1,-2} & = & \frac{k}{k+1},
\end{array}
\]
hence by (\ref{Lrjdef})
\begin{equation}
\label{rbe}
\left[\cR E \right]\ =\ \left[
\begin{array}{cc}
E_{0,0} & E_{0,1} \\
E_{1,0} & E_{1,1}
\end{array}
\right]\ =\ \left[
\begin{array}{cc}
E_k +1 & \frac{2k+1}{k+1} \\
2 E_k & \frac{k+1}{k+2} E_k + 1
\end{array}
\right].
\end{equation}
For $[\cR X]$ it follows from Proposition \ref{xrjcab} that
\begin{equation}
\label{rbx}
\left[\cR X \right]\ =\ \left[
\begin{array}{cc}
X_{0,0} & X_{0,1} \\
X_{1,0} & X_{1,1}
\end{array}
\right]\ =\ \left[
\begin{array}{cc}
k & k E_k^{-1} \\
k+1 & k
\end{array}
\right].
\end{equation}

By way of example we now look at two operators $L \in \K[x]\langle E\rangle$ and compute their associated operators $[\cR L]$ (only the first column) and $L'$. The $2\times 2$ identity matrix is denoted by $I_2$.

\begin{enumerate}
\item $L = (n + 1) E - 2 (2 n + 1)$: Using (\ref{rbe}) and (\ref{rbx}), we obtain
\begin{align*}
[\cR L]e^{(1)}& = \left(([\cR X] + I_2)[\cR E] - 2 (2[\cR X] + I_2)\right)e^{(1)}\\
&= \left[
\begin{array}{c} 
(k+1)(E_k - 1) \\
3(k+1)(E_k - 1)
\end{array}
\right],
\end{align*}
so we can take $L' = E_k - 1$. This comes as no surprise, since $y_n = \sum_{k=0}^\infty \binom{n}{k}^2 = \binom{2n}{n}$ satisfies $Ly=0$.

\item $L = 4 (2 n+3)^2 (4 n+3) E^2 - 2 (4 n+5) \left(20 n^2+50 n+27\right) E + 9 (4 n+7) (n+1)^2$:

\smallskip\noindent
Applying substitution (\ref{substitute}) to $L$, we obtain
\begin{align*}
[\cR L]e^{(1)} &= \left[
\begin{array}{c} 
L_{0,0} \\
L_{1,0}
\end{array}
\right]
\end{align*}
where
\begin{align*}
L_{0,0} \ =\ &\ 4 (2k+3)^2 (4 k+3) E_k^2\\
&+\ \frac{2 \left(592 k^4+1388 k^3+1254 k^2+519 k+81\right)}{k+1}\, E_k\\
&+\ 676 k^3 -889 k^2-466 k-99-(244 k+41) k^2 E_k^{-1}, \\
L_{1,0} \ =\ 
& \,\frac{8 (2 k+3) \left(28 k^3+108 k^2+132 k+51\right)}{k+2}\,  E_k^2\\
& +\ 4 \left(360 k^3+720 k^2+451 k+82\right) E_k \\
& -\ 2 (k+1) \left(74 k^2+377 k+133\right) - 60 (k+1) k^2 E_k^{-1}.\\
\end{align*}
Taking
\[
L'\ =\ {\rm gcrd}\,\left(E_k\, L_{0,0}, E_k\, L_{1,0}\right)
\ =\ E_k-\frac{k+1}{2 (2 k+1)},
\]
we see that $h_k = \frac{1}{\binom{2k}{k}}$ satisfies $L'h = 0$, so
\[
y_n =  \sum_{k=0}^\infty \frac{\binom{n}{k}^2}{\binom{2k}{k}}
\]
is a definite-sum solution of equation $Ly = 0$. \qed
\end{enumerate}
\end{example}

\section{Concluding remarks}

\begin{enumerate}
\item The algorithm \textsc{DefiniteSumSolutions} can be used recursively to produce solutions of $Ly = 0$ in the form of \emph{nested} definite sums.

\item We only used the first column of $[\cR L]$, but this matrix provides information on \emph{all} definite-sum solutions of $Ly = 0$ having the form $y(x) = \sum_{n=0}^\infty c_n P_n(x)$ with $P_n(x) = P_{mk+j}^{(\pi)}(x)$ as given in (\ref{Pmkj}).

\item It seems likely that our approach can be generalized to summands of the form $F(n,k)h_k$ where
\begin{equation}
\label{Fnkgen}
F(n,k) \ =\ \prod_{i=1}^m \binom{a_i n + b_i}{c_i k + d_i}
\end{equation}
with $a_i, c_i \in \N \setminus \{0\}$, $b_i \in \Z$, $d_i \in \N$.

\item Could Galois theory of difference equations help to determine, given $L$, for which values of $m, a_i, b_i, c_i, d_i$ in (\ref{binom}) or (\ref{Fnkgen}) may the associated equation $L'h = 0$ possess nonzero explicitly representable (e.g., Liouvillian) solutions?
\end{enumerate}

%%%%%%%%%%%%%%%%%%%%%%%%%%%%%%%%%%%%%%%%%%%%%%%%%%%%%%%%%%
\section*{Acknowledgements}
%%%%%%%%%%%%%%%%%%%%%%%%%%%%%%%%%%%%%%%%%%%%%%%%%%%%%%%%%%

The author acknowledges financial support from the Slovenian Research Agency (research core funding No.\ P1-0294). The paper was started while he was attending the thematic programme ``Algorithmic and Enumerative Combinatorics'' at the Erwin Schr\"odinger International Institute for Mathematics and Physics in Vienna, Austria. He thanks the Institute for its support and hospitality.

\end{document}